\documentclass [12pt]{article}
\usepackage {a4}
\usepackage [cp1250]{inputenc}
\usepackage [english]{mskrzypczak}

\title{Equational theories of profinite structures}

\author{Michał Skrzypczak\footnote{Author supported by ERC Starting Grant "Sosna"}}

\affil{Institute of Informatics\\University of Warsaw}

\newcommand{\mods}{\ensuremath{\mathbb W\;}\xspace}

\newcommand{\promods}{\ensuremath{\widehat{\mathbb W\;}}\xspace}

\newcommand{\bcap}{\ensuremath{\mathcal B}\xspace}

\newcommand{\reg}{\ensuremath{\textrm reg}}

\newcommand{\forms}{{\ensuremath{\Phi}}\xspace}

\newcommand{\prof}[1]{\widehat{{#1}^\ast}}

\begin{document}

\maketitle

\abstract{In this paper we consider a general way of constructing profinite structures based on a given framework --- a countable family of objects and a countable family of recognisers (e.g. formulas). The main theorem states:
\begin{quote}
A subset of a family of recognisable sets is a lattice if and only if it is definable by a family of profinite equations.
\end{quote}

This result extends Theorem 5.2 from \cite{pin_duality_equational} expressed only for finite words and morphisms to finite monoids.

One of the applications of our theorem is the situation where objects are finite relational structures and recognisers are first order sentences. In that setting a simple characterisation of lattices of first order formulas arise.
}

\section{Introduction}

The following situation is very popular in computer science: the expressive power of a countable set of syntactical objects $\forms$ is studied over a countable family of structures $\mods$.

The following examples illustrate this situation:

\begin{enumerate}
\item $\mods=A^\ast$ and $\forms=\{A^\ast\to M:M\textrm{ is a finite monoid}\}$,
\item $\mods$ are finite trees and $\forms=\{\varphi: \varphi\textrm{ is a first order sentence}\}$,
\item $\mods$ is a set of all finite trees and $\forms$ is a family of all deterministic bottom-up tree automata,
\item $\mods$ are all finite graphs and $\forms$ is the family of all FO or MSO formulas.
\end{enumerate}

One of the natural ideas how to represent \emph{recognition} is to treat a syntactical object $\varphi\in\forms$ as a function $\fun{\varphi}{\mods}{K_\varphi}$ to a finite set $K_\varphi$. Such a function recognises sets of objects of the form
$$\varphi^{-1}(V)\subseteq\mods,$$
for $V\subseteq K_\varphi$. Such sets are usually called regular or recognisable.

All the examples presented above fall into this schema: formula is a function to $\{\bot,\top\}$, homomorfizm is a function to a finite monoid, deterministic automaton maps a tree into the state reached at the root.

In this paper we work with a very general setting of families of objects $\mods$ and of recognisers $\forms$. Using adequate topology on $\mods$ we show how to define a profinite structure $\promods$ extending $\mods$. Moreover we prove that its possible to extend recognisers $\varphi\in\forms$ to all profinite objects $w\in\promods$.

The following theorem is the main result.

\begin{theorem}\label{t:lattices}
A family $\mathcal M$ of recognisable sets of objects is a lattice if and only if it is defined by a set of profinite equations.
\end{theorem}

In paper \cite{pin_duality_equational} authors show analogous theorem in the context of profinite monoids. It is a special case of our result where $\mods=\Sigma^\ast$ and $\forms$ are homomorphisms to finite monoids.

Paper \cite{pin_topological_recognition} is devoted to the idea of recognisers --- particular functions defined on a given topological (or more generally uniform) space. Authors show that each Boolean algebra of subsets of a space has a \emph{minimal recogniser}. Additionally various additional structures of the space (e.g. the structure of a monoid) provides additional properties of a recogniser. The key tool is the Stone-Priestley duality.

The approach presented in this paper is different. We start with a fixed family of recognisers and using them we provide adequate topology on $\mods$. Using this topology we extend the space to a profinite structure and study the particular algebra of recognisable sets.

\section{Profinite structures}

\begin{definition}
A \emph{framework} consists of
\begin{enumerate}
\item a countable set $\mods$ of objects,
\item a countable set $\forms$ of recognisers, i.e. functions $\fun{\varphi}{\mods}{K_\varphi}$ where $K_\varphi$ is a finite set.
\end{enumerate}

Additionally two properties must hold:

\begin{description}
\item[a)] Each object $w\in\mods$ is totally described by some recogniser. That is for every object $w\in\mods$ there is some recogniser $\varphi\in\forms$ such that $\varphi(w)\neq\varphi(w')$ for $w'\neq w$.
\item[b)] Recognisers are closed under intersections. That is for every recognisers $\varphi_1,\varphi_2\in\forms$ and every sets of values $V_1\in K_{\varphi_1}, V_2\in K_{\varphi_2}$ there exists a recogniser $\varphi\in\forms$ and a set $V\subseteq K_\varphi$ such that
$$\forall_{w\in \mods}\  \left(\varphi_1(w)\in V_1\wedge \varphi_2(w)\in V_2\right)\Leftrightarrow \varphi(w)\in V.$$
\end{description}
\end{definition}

It is easy to check that all the examples from the introduction satisfy both axioms so they form \emph{frameworks}.

Fix a framework $(\mods, \forms)$.

\begin{definition}
A language $L\subseteq \mods$ is called \emph{regular} or \emph{recognisable} if there exists $\varphi\in\forms$ and $V\subseteq K_\varphi$ such that
$$L=\varphi^{-1}(V).$$

The family of all regular languages is denoted as $\reg(\mods)$.
\end{definition}

\begin{observation}
$\reg(\mods)$ is a Boolean algebra.
\end{observation}

\begin{proof}
By Property b) regular languages are closed under intersection. Of course they are also closed under complementation. So they form a Boolean algebra.
\end{proof}

List all recognisers in $\forms$ in a sequence $\forms=\{\varphi_0,\varphi_1,\ldots\}$. It is good to think that recognisers appearing further in the sequence are \emph{more complicated}. Let $K_i=K_{\varphi_i}$ and let $X=\prod_{i\in\N} K_i$. Because $K_i$'s are finite,  $X$ is a homeomorphic copy of Cantors discontinuum. Let $\fun{\mu}{\mods}{X}$ be defined as follows
$$\mu(w)=\left(\varphi_0(w), \varphi_1(w),\varphi_2(w),\ldots\right).$$
In other words $\mu$ maps an object $w\in\mods$ to a sequence of values of all recognisers on that object.

It is easy to see that $\mu$ is 1-1 because of Property a).

\begin{definition}
Let $\promods = \overline{\mu(\mods)}\subseteq X$. The elements of the set $\promods$ are called \emph{profinite objects} of the framework $(\mods, \forms)$. The topology on $\promods$ is defined as a topology induced from $X$.
\end{definition}

Since the order of coordinates in Cantors discontinuum does not affect its topology, we obtain the following fact.

\begin{fact}
The construction described above does not depend on the order of recognisers in the sequence $\varphi_0,\varphi_1,\ldots$.
\end{fact}

The following proposition summarises the properties of $\promods$.

\begin{proposition}
\mbox{}
\begin{enumerate}
\item $\mods$ naturally embeds into $\promods$,
\item $\promods$ is a compact topological space,
\item for each $w\in\mods$ the point $\mu(w)\in\promods$ is isolated,
\item $\mu(\mods)$ is a countable dense subset of $\promods$.
\end{enumerate}
\end{proposition}

\begin{proof}
\mbox{}
\begin{enumerate}
\item The natural embedding is $\mu$.
\item $\promods$ is a closed subset of a compact space, so it is compact.
\item This is an easy consequence of property a) of the framework.
\item  $\promods$ is a closure of $\mu(\mods)$, so $\mu(\mods)$ is dense in $\promods$.
\end{enumerate}
\end{proof}

Note that recognisers naturally extend to $\promods$.

\begin{proposition}
For each recogniser $\varphi_i\in\forms$ we can extend it to all profinite objects $w\in \promods\subseteq\prod K_i$ by an equation
$$\varphi_i(w)= w_i\in K_i.$$

For $w\in \mods$ this definition is consistent with the original one.
\end{proposition}

The following lemma enables us to define regular languages from the topological point of view.

\begin{lemma}\label{l:duality}
A language $L\subseteq \mods$ is regular if and only if $\overline{L}$ is an clopen subset of $\promods$.

Moreover $L\mapsto \overline{L}$ is an isomorphism of the Boolean algebra of regular languages in $\mods$ and the Boolean algebra of clopen subsets of $\promods$.
\end{lemma}

\begin{proof}
Regular languages form a clopen base of the topology and they are closed under finite Boolean operations. In a compact space each clopen subset is a Boolean combination of the clopen base sets.
\end{proof}

\section{Topology}

In this section we provide a general characterisation of sublattices and Boolean subalgebras of the algebra of clopen subsets of a compact space. This whole theory is based only on topological properties and holds for any compact space.

Fix $X$ to be a compact space and let $\bcap$ denote the Boolean algebra of clopen subsets of $X$.

\begin{definition}
A sublattice of $\bcap$ is any subset $\mathcal M\subseteq \bcap$ closed under union and intersection and containing $\emptyset$ and $X$.
\end{definition}

\begin{definition}\label{d:eq_top}
An \emph{equation} is a formula $u\to v$ for $u,v\in X$. We say that a subset $A\subseteq X$ \emph{satisfies} an equation $u\to v$ if the following property holds
$$u\in A\Rightarrow v\in A.$$

For a given set of equations $\mathcal E\subseteq X^2$, the family of subsets that satisfy all equations in $\mathcal E$ will be called a family of subsets defined by $\mathcal E$.
\end{definition}

\begin{theorem}\label{t:lattice_top}
For a given compact space $X$ a family $\mathcal M\subseteq \bcap(X)$ is a lattice if and only if it is defined by some set of equations.
\end{theorem}

\begin{proof}
Of course for a given set of equations, a family of clopen subsets satisfying them is a lattice.

Take any lattice $\mathcal M\subseteq\bcap$. Let $\mathcal E$ be a set of all equations satisfied by every set in $\mathcal M$. Take any subset $A\in\bcap$ satisfying $\mathcal E$. We will show that $A\in\mathcal M$.

Let $$\mathcal U=\left\{M\in\mathcal M: M\subseteq A\right\}.$$ If $\bigcup \mathcal U= A$, then $\mathcal U$ is a covering by open sets of a compact $A$, so there is a finite family $M_1,M_2,\ldots,M_n\in \mathcal U$ such that $\bigcup M_i= A$. So $A\in\mathcal M$. Assume by contradiction that there exists $x\in A\setminus \bigcup \mathcal U$.

Consider $$\mathcal V=\left\{M\in\mathcal M: x\in M\right\}.$$ If $\bigcap \mathcal V\subseteq A$, then complements of elements in $\mathcal V$ cover $X\setminus A$ so there is a finite family $M_1,M_2,\ldots,M_n\in\mathcal V$ such that $\bigcap M_i\subseteq A$. But then $x\in\bigcap M_i\subseteq A$ and $\bigcap M_i\in\mathcal M$, so a contradiction to the fact that $x\notin \bigcup \mathcal U$. So there exists $y\in\bigcap \mathcal V\setminus A$.

Consider equation $x\to y$ and any $N\in\mathcal M$ such that $x\in N$. Then $N\in\mathcal V$ so by the definition $y\in N$ so $N$ satisfies $x\to y$. Therefore $(x\to y)\in\mathcal E$. But $A$ does not satisfy $x\to y$. A contradiction.
\end{proof}

\begin{corollary}
A family $\mathcal M\subseteq \bcap$ is a Boolean subalgebra if and only if it is defined by a symmetric set of equations $u\leftrightarrow v\equiv u\to v,v\to u$.
\end{corollary}

\section{Main theorem}

Using the topological result from the previous section we can prove the main theorem from the introduction.

\begin{theorem}
A family $\mathcal M$ of recognisable sets of objects is a lattice if and only if it is defined by a set of profinite equations.
\end{theorem}

Firstly we adopt the definition of equation to the case of regular languages. This definition follows the one from \cite{pin_duality_equational}.

\begin{definition}\label{d:eq_reg}
We say that a regular language $L\subseteq \mods$ satisfies an equation $u\to v$ for $u,v\in \promods$ if
$$u\in\overline{L}\ \Rightarrow\ v\in\overline{L}.$$
\end{definition}

In other words language $L$ recognised by a recogniser $\varphi$ and a set of values $K\subseteq K_\varphi$ satisfies $u\to v$ iff
$$\varphi(u)\in K\ \Rightarrow\ \varphi(v)\in K.$$

\begin{proof}[Proof of Theorem \ref{t:lattices}]
Take any subfamily of regular languages $\mathcal M\subseteq \reg(\mods)$ and consider
$$\overline{\mathcal M}=\{\overline{L}\subseteq \promods:L\in\mathcal M\}.$$

By Lemma \ref{l:duality} $\overline{\mathcal M}\subseteq \bcap(\promods)$ and $\mathcal M$ is a lattice if and only if $\overline{\mathcal M}$ is a lattice. Additionally $\mathcal M$ is defined by a set of equations in the meaning of Definition \ref{d:eq_reg} if and only if $\overline{\mathcal M}$ is defined by a set of equations in the meaning of Definition \ref{d:eq_top}.

But $\overline{\mathcal M}$ is defined by a set of equations of and only if it is a lattice, because of Theorem \ref{t:lattice_top} and the fact that $\promods$ is a compact topological space.
\end{proof}

\begin{corollary}
A family $\mathcal M$ of recognisable languages of objects is a Boolean algebra if and only if it is defined by a set of profinite symmetric equations.
\end{corollary}

\section{Conclusions}

In this section we propose various applications of the main theorem.

Firstly consider a case when $\mods=\Sigma^\ast$ and $\forms$ consists of all homomorphisms of $\Sigma^\ast$ to finite monoids. In that case $\promods$ is just a space of profinite words $\prof{\Sigma}$ and Theorem \ref{t:lattices} coincides with Theorem 5.2 from \cite{pin_duality_equational}.

It turns out that Theorem \ref{t:lattices} gives some insight into the structure of lattices of first order formulas. Fix a relational signature $\Sigma$ and consider $\mods$ as a family of all finite structures over $\Sigma$. Let $\forms$ be a set of all FO sentences over $\Sigma$. Of course $\mods,\forms$ is a framework. We denote this framework as \emph{first order framework} over signature $\Sigma$.

\begin{lemma}\label{l:compactness}
If $\mods, \forms$ is a first order framework then for every profinite object $w\in\promods$ there exists a structure $S(w)$ over $\Sigma$ such that for all FO sentences $\varphi\in\forms$
$$\varphi(w)=\top\ \Leftrightarrow\ S(w)\models \varphi.$$
\end{lemma}

\begin{proof}
An easy application of the compactness theorem.
\end{proof}

Of course there are infinite structures that don't represent any element of $\promods$. For example take $\Sigma=\{\leq\}$ and let $\varphi$ express that $\leq$ is a linear order and that $\forall_x\exists_y\ x<y$. Then $(\omega,\leq)$ is a model for $\varphi$ but no finite structure satisfies it. So $(\omega,\leq)$ does not represent any element of $\promods$.

\begin{theorem}
Every lattice of first order formulas $\mathcal M\subseteq FO(\Sigma)$ is defined by some family of implications $(u\to v)$, where $u,v$ are (potentially infinite) $\Sigma$-structures.

We say that a formula $\varphi$ satisfies an implication $(u\to v)$  iff
$$(u\models\varphi)\ \Rightarrow \ (v\models\varphi).$$
\end{theorem}

\begin{proof}
By \ref{t:lattices} a lattice of formulas is defined by some family of profinite equations. But using \ref{l:compactness} each profinite object $w\in\promods$ can be interpreted as a structure $S(w)$. Therefore, we can put structures instead of profinite objects in all equations.
\end{proof}

\subsection{Acknowledgements}

The author would like to thank Mikołaj Bojańczyk, Damian Niwiński and Henryk Michalewski for their helpful comments.

\bibliography{mskrzypczak}

\end{document}